\documentclass[12pt,onecolumn]{article}
\usepackage[a4paper, margin=1in]{geometry}
\usepackage[T1]{fontenc}
\usepackage{lmodern}
\usepackage{appendix}
\usepackage{authblk}
\usepackage{amsmath,amsthm} 
\usepackage{amssymb}  
\usepackage{xcolor,mathrsfs}
\usepackage{%
	algorithm,%
    algpseudocode,
	bbm,%
    dsfont,
	enumerate,%
	mathtools,%
	paralist,%
	pgf,%
	pgfplots,%
	subfigure,%
	tikz,%
	tikzscale,%
    url,
}
\algnewcommand{\TRUE}{\textbf{true}}
\algnewcommand{\FALSE}{\textbf{false}}
\usepackage{pgfplots}
\usepgfplotslibrary{fillbetween}
\usetikzlibrary{patterns}
\newcommand{\x}{\textbf{x}}
\newcommand{\y}{\textbf{y}}
\newcommand{\p}{\textbf{p}}

\newtheorem{thm}{Theorem}
\newtheorem{prop}[thm]{Proposition}
\newtheorem{cor}[thm]{Corollary}

\newtheorem{lem}[thm]{Lemma}

\newtheorem{defn}{Definition}

\newtheorem{exmp}{Example}

\title{\LARGE \bf
How Irrationality Shapes Nash Equilibria:\\
A Prospect-Theoretic Perspective\footnote{This work was supported by D\'efi Inria-EDF.\\ Email addresses of authors: {\tt\small \{ashok-krishnan.komalan-sindhu, helene.le-cadre, ana.busic\} @inria.fr}}}

\author[1,2]{Ashok Krishnan K.S.}
\affil[1]{Inria,
Paris, France.}
\affil[2]{DI ENS, École Normale Supérieure, PSL Research University,
Paris, France.}
\author[3]{H\'el\`ene Le Cadre}
\affil[3]{ Inria, Univ. Lille, CNRS, Centrale Lille, UMR 9189 CRIStAL, F-59000, Lille, France.}
\author[1,2]{Ana Bu\v si\'c}
\date{}
\begin{document}
\maketitle
\begin{abstract}
Noncooperative games with uncertain payoffs have been classically studied under the expected-utility theory framework, which relies on the strong assumption that agents behave rationally. However, simple experiments on human decision makers found them to be not fully rational, due to their subjective risk perception. Prospect theory was proposed as an empirically-grounded model to incorporate irrational behaviours into game-theoretic models. But, how prospect theory shapes the set of Nash equilibria when considering irrational agents, is still poorly understood. To this end, we study how prospect theoretic transformations may generate new equilibria while eliminating existing ones. Focusing on aggregative games, we show that capturing users' irrationality can preserve symmetric equilibria while causing the vanishing of  asymmetric equilibria. Further, there exist value functions which map uncountable sets of equilibria in the expected-utility maximization framework to finite sets. This last result may shape some equilibrium selection theories for human-in-the-loop systems where computing a single equilibrium is insufficient and comparison of equilibria is needed. 
\end{abstract}

\section{Introduction}

Game theory aims to model and analyze the outcomes of interactions between strategic agents in competition. Game theory has been applied to diverse fields, most prominently in  economics~\cite{aumann1992handbook,samuelson2016game}. The most commonly used paradigm in game theory is the von Neumann-Morgenstern expected utility (EUT) maximization~\cite{von1944theory}, which considers rational agents working with expected utilities. Under the EUT paradigm, when faced with a choice between different prospects (each of which is a set of outcomes with an associated distribution), agents pick the one that leads to the maximization of  expected utility. This leads to analytically simple and tractable models, which have proven to be very successful in theory~\cite{schoemaker1982expected}. However, a number of experiments with human participants, revealed that very often, humans are not fully rational~\cite{jallais2005allais}; for example, participants in the experiments favoured prospects which were certain, over prospects which were uncertain, but had higher expected utility. Deviations from the standard EUT model are labelled as \emph{bounded rational}~in opposition to the fully rational EUT setting~\cite{grune2007bounded}. In order to place humans at the heart of decision systems, one needs to rely on bounded rationality, in order to understand how irrationality impacts their decision process~\cite{gan2022application}.

Prospect theory (PT)~\cite{kai1979prospect,tversky1992advances} is a theory at the intersection of behavioural economics, decision theory and psychology, used to model humans' subjective perception of risk under uncertainty. A number of mathematical models  are available for  PT type behaviour~\cite{stott2006cumulative}. However, the application of prospect theory in games is still in its early stages. 

\textit{Literature Review}: 
 There have been a few works which used prospect theoretic tools in game theory. In \cite{shalev2000loss}, existence results for equilibrium are obtained, under a linear PT transformation function. Other works that incorporate PT in finite games and provide existence results include \cite{vahid2019modeling},~\cite{merrick2016modeling},~\cite{metzger2019non}~and \cite{keskin2016equilibrium}. In \cite{sanjab2017prospect}, prospect theory is used to study the effects of subjective perceptions in the case of a zero-sum security game which arises in a drone-based delivery system, by means of simulations. There have been few works that apply PT in general games where players have continuous strategy sets. An application of prospect theory to prosumer games in an electrical grid is done in \cite{el2017managing}. They identify conditions under which one can guarantee existence of equilibria, under a specific type of PT transformation. In \cite{li2014users}, a wireless random access game with two players with prospect theoretic preferences is studied, to obtain existence and structural results. Existence of PT-Nash equilibria for nonconvex nonsmooth PT preferences is presented in \cite{fochesato2025noncooperative}, along with an algorithm that converges to a critical Nash equilibrium with convergence guarantees. However, how irrationality impacts Nash equilibria solutions of an EUT noncooperative game is not yet well understood. 

 \textit{Our Contributions}:
 \begin{enumerate}
     \item We show that for concave PT value functions, in aggregative coordination games, the presence of irrational players leads to a vanishing of asymmetric Nash equilibria, while symmetric equilibria are preserved.
     \item We show that there exist PT value functions which can transform the infinite set of equilibria of the EUT noncooperative game to a finite set.
     \item We apply our results to an electricity market. We show numerically that considering homogeneous end users, irrationality leads to an increase in the tariffs charged to the end users. Further, if the selection function coincides with the maximization of the economic efficiency, irrationality may lead to a decrease of the tariffs.
 \end{enumerate}
 
\textit{Notation}:
To denote the empty set we use $\phi$, and $\mathds{1}$ for the indicator function. A vector of all ones is $\textbf{1}$. By $\textbf{z}_{-i}$ we mean the vector $\textbf{z}$ with the $i$-th component removed. By $\mathbb{E}_{d}[\cdot]$, we mean the expectation with respect to the distribution $d$ and $\mathrm{proj}_C[\cdot]$  is the projection onto a convex set $C$.

\section{Overview of Prospect Theory}
We first illustrate the framework of PT  using  terminology from the PT literature~ \cite{kai1979prospect},~\cite{stott2006cumulative}. Consider a player obtaining decreasing (in order of attractiveness) random outcomes $Z_j$ with probability $q_j$, for $j=1,...,M$, where $\sum_j q_j=1$. The collection $g=(Z_j,q_j)_{j=1}^M$ is called a \emph{prospect}. Under the expected utility paradigm, whenever a player is faced with the choice between two prospects $g_1=(Z_j(1),q_j(1))_{j=1}^M$ and $g_2=(Z_j(2),q_j(2))_{j=1}^M$, they choose the prospect that maximizes the expected utility; i.e., choose $g_1$ if
$
    \sum_{i=1}^MZ_j(1)q_j(1)\ge \sum_{i=1}^MZ_j(2)q_j(2);
$
pick $g_2$ otherwise. Under the prospect theoretic framework, however, the perception of outcomes changes with respect to a reference point. Outcomes above this reference are seen as gains, and those below as  losses. When faced with gains, agents will behave in a \emph{risk averse} manner, and when faced with losses, agents tend to be \emph{risk seeking} or \emph{risk neutral}. This perception is captured by an outcome to value map $v$, which, in order to model the above (and other) aspects of bounded rational behaviour, is assumed in general to be convex below a threshold, concave above it.  One such function is depicted in Fig. \ref{fig:distortion-curve-val}. Here the agents are risk averse for gains and risk neutral for losses, with the  zero being the reference point that demarcates gains from losses.
\begin{figure}
    \centering
    \begin{tikzpicture}[scale=1]
\begin{axis}[
xlabel=outcome,ylabel=value,
          xmax=5,ymax=4,
          xmin=-2,
          axis lines=middle,ticks=none]
\addplot[smooth,black,mark=none,
line width=1.5pt,domain=0:9.5,
samples=63]  {(1/log2(2.712))*log2(1+x)};
\addplot[smooth,black,mark=none,
line width=1.5pt,domain=-2:0,
samples=63]  {x};
\addplot[dotted,black,mark=none,
line width=1.5pt,domain=0:6,
samples=63]  {x};
\node at (axis cs:3,1.6) {$v$};
\end{axis}
\end{tikzpicture}
    \caption{Example of a prospect theoretic value function $v$ that maps outcomes to values; dotted line has unit slope.}
    \label{fig:distortion-curve-val}
\end{figure}

In addition to the notion of value relative to risk, the weight accorded to outcomes is modified as well. The modified value of the prospect is
$$
    V(g)=\sum_{j=1}^M \tilde q_jv(Z_j),
$$
where $v$ is a function similar to the one in  Fig.~\ref{fig:distortion-curve-val}, and $\tilde q_j$ are the weights for the distorted outcomes, given by~\cite{stott2006cumulative},
$$
    \tilde q_1 =\pi(q_1),
    \tilde q_j = \pi(\sum_{m=1}^j q_m)-\pi(\sum_{m=1}^{j-1} q_m),~j>1,
$$
for an appropriate weighting function $\pi$, which is monotone increasing and bounded between 0 and 1. This models  the overweighting of small probabilities, and the underweighting of large probabilities, for large outcomes. Faced with a choice between two prospects $g_1$ and $g_2$, the agent will make a decision comparing their PT modified values $V(g_1)$ and $V(g_2)$, and choosing the better option. It is clear that these decisions are likely to be very different from those predicted by expected utility theory.

\section{Game Theoretic Framework}\label{sec:GameModel}
 The set of players in the game is denoted as $\mathcal N := \{1,...,N,N+1\}$. Player $N+1$ (referred to as the coordinator,) coordinates between the other $N$ players, and interacts with them via its strategy $\p\in\mathcal P$. The other players ($1$ to $N$), choose strategies $x_i\in\mathcal X_i$, for $i=1,...,N$, respectively. We will assume that $\mathcal X_i$ are compact and convex. Let $\mathcal X=\prod_{i=1}^N\mathcal X_i$ denote the joint feasibility set. For player $i$, let $J_i(x_i,\x_{-i},\p,\xi)$ denote the outcome received by player $i$ when the coordinator plays $\p$, other players have joint strategy $\x_{-i}$, and the common randomness is $\xi$. In particular, we assume that the outcomes have the form
 \begin{align}\label{eq:defn-pre-utility}
		J_i(x_i,\x_{-i},\p,\xi)=-a_iU(\x)+b_i-C_i(p_i,x_i,\xi),
	\end{align}
    where $U(\x)$ denotes an aggregative common usage benefit,
    \begin{align}
		U(\x)=\left(\frac{1}{N}\sum_{j=1}^Nx_j-\frac{1}{N}\sum_{j=1}^Ny_j\right)^{2},
	\end{align} 
    and $C_i(p_i,x_i,\xi)$ is an individual cost,
   \begin{align}
		C_i(p_i,x_i,\xi)=p_ix_i+(y_i-x_i)\xi,
	\end{align} 
    and $a_i>0,b_i\ge 0$.  The outcomes represent a game where players gain by getting their collective strategy $\x$ close to some desired collective target $\y$. The cost may be interpreted as the individual cost for not meeting the individual target $y_i$. 
    We will assume that the random variable $\xi$ takes values over the set $\Xi=\{\xi_1,...,\xi_M\}$ with distribution $q=(q_1,...,q_M)$, 
$
        \overline{\xi}=\mathbb{E}_q[\xi],
$
and positive variance.    
 The players 1 to $N$  play a subgame based on utilities derived from the outcomes. We have two games here, the EUT and PT games; the first  follows the expected utility paradigm and the second uses prospect theory. 

\subsection{The EUT  Game Model}\label{sec:GameModel-EUT}
We define the EUT game as the tuple 
$$\mathcal G_{\text{EUT}}=(\mathcal N,(\overline J_i)_{i\in \mathcal N},\mathcal P, (\mathcal X_i)_{i\in \mathcal N}),$$
with utilities 
$$\overline J_i=\mathbb{E}_q[J_i(x_i,\x_{-i},\p,\xi)].$$
For this game, we define a Nash equilibrium for the subgame between the players $1,...,N$, for a fixed $\p$ as follows.
\begin{defn}
    For $\p\in\mathcal P$, a joint strategy $\x^*=(x_i^*)_{i=1}^N$ is a Nash equilibrium for $\mathcal G_{EUT}$ if
    \begin{align}\label{def:Nash-eqbm-EUT}
    \overline J_i(x_i^*,\x_{-i}^*,\p)\ge  \overline J_i(x_i,\x^*_{-i},\p)~\forall x_i\in\mathcal X_i~,i=1,...,N.
\end{align}
\end{defn}

Thus $\p$ can be considered an incentive from the coordinator, which results in  player response $\x^*(\p)$. Since we assume that the players are playing a coordination game, it is likely that there are multiple Nash equilibria~\cite{cooper1999coordination}. We denote the set of all Nash equilibria at $\p$ by $X^*(\p)$. The coordinator picks one of the equilibria by means of an optimization, 
\begin{align}
    \hat\x^{*}(\p)=\arg_{\x^*(\p)\in X^*(\p)}\max\mathscr F(\x^*(\p)),
\end{align}
with $\mathscr{F}(\cdot)$ being a pre-specified selection function. The coordinator then drives the system to a particular behaviour by an appropriate choice of $\p$, by optimizing some function   $\overline J_0(\p,\hat\x^{*}(\p))$. We present two possible candidates for this function in Section \ref{sec:numerical-study}.

\subsection{PT Game Model}\label{sec:PT-framework-intro}

We define the PT  game as the tuple
$$\mathcal G_{\text{PT}}=(\mathcal N,(\widetilde J_i)_{i\in\mathcal N},\mathcal P, (\mathcal X_i)_{i\in\mathcal N}),$$
where the utilities  are
$$
    \widetilde J_i(x_i,\x_{-i},\p)=\mathbb{E}_q[v_i\circ J_i(x_i,\x_{-i},\p,\xi_j)],
$$
where $v_i$  are concave, increasing and differentiable, and $v_{N+1}(z)=z$.

An example of a $v_i$ that satisfies this assumption is
\begin{align}\label{eq:log-v-function}
    v_i(z)=
         \log(1+z)\mathds{1}_{\{z\ge 0\}}+z\mathds{1}_{\{z<0\}},~i=1,...,N.
     \end{align}
The assumption on $v_{N+1}$ ensures that the behaviour of the coordinator does not change. Thus, the PT game corresponds to a setting where the player \emph{bounded rationality} is taken into account, while the coordinator is always considered as a rational player. The coordinator may not be able to depart from rational behaviour owing to ethical or legal considerations. 
This is true, for example, when the coordinator represents an industry or governmental agency with system-based objective functions, and the other players are humans with subjective risk preferences.
\begin{defn}
    For $\p\in\mathcal P$, a joint strategy $\x^*=(x_i^*)_{i=1}^N$ is a Nash equilibrium for the  PT game if
    \begin{align}\label{def:Nash-eqbm-PT}
    \widetilde J_i(x_i^*,\x_{-i}^*,\p)\ge  \widetilde J_i(x_i,\x^*_{-i},\p)~\forall x_i\in\mathcal X_i~,i=1,...,N.
\end{align}
\end{defn}

The set of Nash equilibria in the PT framework is denoted by $\widetilde X^*(\p)$. Since multiple equilibria are possible in the PT setting, we use the same selection function $\mathscr F(\cdot)$ as in the EUT case, and the coordinator takes $\overline J_0$ as objective function.

\section{The EUT Game: Structure and Solutions}
The game $\mathcal G_{\text{EUT}}$ yields multiple Nash equilibria at certain strategies $\p$ of the coordinator, as seen below.
\begin{lem}\label{lem:Nash-eq-set-EUT}
     We have
     \begin{align}\label{eqn:equilibrium-criterion}
         X^*(\p)=
             \left\{\x\in\mathcal X:\sum_{j=1}^Nx_j=\sum_{j=1}^Ny_j-\frac{N^2}{2}\kappa_{\p}\right\}\mathds{1}_{\{\p\in\mathcal P_1\}}
     \end{align}
     where $$\kappa_{\p}=\frac{p_1-\overline\xi}{a_1},~
         \mathcal P_1=\left\{\p\in\mathcal P:\frac{p_i-\overline\xi}{a_i}=\frac{p_1-\overline\xi}{a_1},~ \forall i\right\}.
    $$
 \end{lem}
\begin{proof}
    For a fixed $\p$, the Nash equilibrium~\eqref{def:Nash-eqbm-EUT} is given by any point $\x^*=(x_i^*)_{i=1}^N$ that satisfies
	\begin{align}\label{eqn:Nash-condtn-loc}
		x_i^{*}\in\arg_{x_i}\max \overline J_i(x_i,\x_{-i}^*,\p),~i=1,...,N.
	\end{align}
    Since $\frac{\partial^2}{\partial x_i^2}\overline J_i(x_i,\x_{-i},\p)<0$, it follows that any $\x^*$ that satisfies
$$
     \frac{\partial}{\partial x_i}\overline J_i(x_i^*,\x_{-i}^*,\p) =0,~i=1,..,N,$$
     
 is a Nash equilibrium. This is equivalent to the family of equations,
 \begin{align}\label{eqn:condtn-1}
    \frac{N^2}{2a_i}(p_i-\overline\xi)= (\sum_{j=1}^Ny_j-\sum_{j=1}^Nx_j),~i=1,...,N.
 \end{align}
If $\p\notin\mathcal P_1$, the system of equations \eqref{eqn:condtn-1} is inconsistent and there is no solution. If $\p\in\mathcal P_1$, we can see that the system \eqref{eqn:condtn-1} is equivalent to the single equation
\begin{align}\label{eqn:simplified-x-y-p}
         \sum_{j=1}^Nx_j-\sum_{j=1}^Ny_j=-\frac{N^2}{2}\kappa_{\p}.
     \end{align}
\end{proof}
At any $\p\in\mathcal P_1$ there are uncountably many Nash equilibria. The coordinator has to select one of these equilibria.

\textit{Selection of Equilibrium}: To select an equilibrium, we use a variant of Jain's fairness index~\cite{jain1984quantitative}. This criterion chooses a \lq fair\rq~ equilibrium. For $N$ different players being allocated $\delta_i$ units each of some resource, Jain's index is given by
 \begin{align}
     \mathcal{J}(\delta_1,...,\delta_N):=\frac{1+(\sum_i \delta_i)^2}{1+N\sum_i\delta_i^2}.
 \end{align}
If the aggregator seeks fairness in terms of the deviation from planned consumption for each customer, i.e., $\delta_i\triangleq x_i-y_i$, it selects the element $\hat\x^{*}$ such that
\begin{align}\label{eqn:opt-jain}
    \hat\x^{*}(\p) =\arg_{\x\in X^*(\p)}\max \mathcal{J}(x_1-y_1,...,x_N-y_N).
\end{align}
For $\p\in\mathcal P_1$ and $\x\in X^*(\p)$, note that $\sum_i\delta_i=-\frac{N^2}{2}\kappa_{\p}$. Hence \eqref{eqn:opt-jain} is equivalent to
 \begin{align}
     \min \sum_{j=1}^N (x_j-y_j)^2~s.t.~\x\in X^*(\p),
 \end{align}
 which is solved by
 \begin{align}\label{eq:selected-eq}
     \hat x_i^{*}=y_i-\frac{N(p_i-\overline{\xi})}{2a_i},~\forall i.
 \end{align}
 This equilibrium requires consumers to  have information only about their own parameters.

 \textit{Convergence to Equilibrium}:
It is easy to see that the game between the $N$ players is a weighted potential game~\cite{monderer1996potential}.
 \begin{lem}
    For any price $\p\in\mathcal{P}$, the game between the $N$ players is a weighted potential game, i.e., for all $i$, and all $\x=(x_i,\x_{-i}),\tilde\x=(x_i^{\prime},\x_{-i})$, 
     \begin{align*}
        \overline J_i(x_i,\x_{-i},\p)-\overline J_i(x_i^{\prime},\x_{-i},\p)=a_i(\Phi_{\p}(\x)-\Phi_{\p}(\tilde\x)),
     \end{align*}
     where
     $$
         \Phi_{\p}(\x)=-\frac{(\sum_{j=1}^Nx_j-\sum_{j=1}^Ny_j)^2}{N^2}-\kappa_{\p}\sum_{j=1}^Nx_j.
     $$
\end{lem}

 We also see that
$
X^*(\p)=\arg_{\x}\max\Phi_{\p}(\x).
$
 The potential structure inspires the following gradient play algorithm, Algorithm \ref{alg:con-alg-grad}, for converging to Nash equilibria of the EUT game, for a fixed $\p$. The number of iterations $\hat N$ can be chosen based on appropriate convergence criteria. Nash equilibria obtained by this algorithm are then passed on to the coordinator, which chooses one of the equilibria based on the fairness criterion, and then updates the price. This update step is shown in Algorithm \ref{alg:sel-alg-prix}.
 
 \begin{algorithm}
  \caption{Distributed gradient play for learning Nash equilibria}
  \begin{algorithmic}[1]
  \State $n=1,\epsilon_n>0,\x=\x(0)$
  \While {$n\le \hat N$}
      \State $i\gets n\mod N$
      \State $\Delta_i \gets -\frac{2a_i}{N^2}(\sum_jx_j-\sum_jy_j)-p_i+\overline\xi$
      \State $x_i \gets x_i+\epsilon_n \Delta_i$
      \State $n \gets n+1$
    \EndWhile
  \end{algorithmic}\label{alg:con-alg-grad}
\end{algorithm}

\begin{algorithm}
  \caption{Coordinator's update for $\p$}
  \begin{algorithmic}[1]
  \State $t=0,\p=\p(0),\delta>0$
  \While {$t\ge 0$}
    \State Players $1,..,N$ do algorithm  \ref{alg:con-alg-grad} $N_0$ times to obtain subset of equilibria $\hat X^*(\p)$
      \State Select $\hat\x(\p)=\arg_{\x^*(\p)\in \hat X^*(\p)}\max\mathscr F(\x^*(\p))$
      \State $\p\gets\p+\delta\frac{\partial J_0(\x,\p)}{\partial \p}|_{\x=\hat\x(\p)}$
      \State $t\gets t+1$
    \EndWhile
  \end{algorithmic}\label{alg:sel-alg-prix}
\end{algorithm}
Due to the potential structure of the game, it follows from \cite[Theorem 1]{ermoliev2002repeated} that  $\x(n)$ in Algorithm \ref{alg:con-alg-grad} converges to some point in $X^*(\p)$, as $n\to\infty$. Note that Algorithm \ref{alg:con-alg-grad} uses global information of the actions of all players, but does not need information about other players' utilities. Thus it can be implemented as a distributed algorithm with information sharing.

\section{The PT Game: Preservation and Vanishing of Equilibria}
We have the PT utilities,
$
    \widetilde J_i(x_i,\x_{-i},\p).
$
As in the EUT case, we consider the game  for a fixed $\p$. Since $\widetilde J_i$ is concave in $x_i$, the first derivative being zero is necessary and sufficient for equilibrium. Setting $\frac{\partial \widetilde J_i}{\partial x_i}=0$  is equivalent to the fixed point equation
$$\x=G_{\p}(\x),$$
where $G_{\p}(\x)=[G_{\p,1}(\x) \cdots G_{\p,N}(\x)]$, with
\begin{align*}
    G_{\p,i}(\x)\triangleq \sum_{j=1}^Ny_j-\sum_{j=1,j\ne i}^Nx_j-\frac{N^2}{2a_i}(p_i-\mathbb{E}[\tilde\xi^{\x,\p}(i)]),
\end{align*}
where $\tilde\xi^{\x,\p}(i)$ is a random variable over $\Xi$ with distribution
\begin{align}\label{eq:transf-prob-dist}
    \mathbb{P}[\tilde\xi^{\x,\p}(i)=\xi_k]=\frac{q_kv_i^{\prime}(J_i(x_i,\x_{-i},\p,\xi_k))}{\sum_{j=1}^Mq_jv_i^{\prime}(J_i(x_i,\x_{-i},\p,\xi_j))}.
\end{align}
In the case of the EUT game, we had $\tilde\xi^{\x,\p}(i)$  replaced by $\xi$. The impact of the PT modification is to shift the mean of the randomness.  However, the dependence of this shift on each strategy $\x$,  price $\p$ and player $i$ makes it much harder to obtain an explicit form for the equilibrium, unlike what we had in the EUT setting. 
\begin{lem}\label{lem:charac-eq-PT}
    A joint strategy $\x$ is a Nash equilibrium for the PT game at price $\p$ if and only if there exists $\kappa_{\p,\x}$ that satisfies,
    \begin{align}\label{eqn:constr-coupl-PT}
        \sum_{j=1}^Nx_j-\sum_{j=1}^Ny_j=-\frac{N^2}{2}\kappa_{\p,\x},
    \end{align}
    and 
    \begin{align}\label{eq:defn-kappa-p-x}
      \kappa_{\p,\x}=\frac{p_i-\mathbb{E}[\tilde\xi^{\x,\p}(i)]}{a_i}= \frac{p_1-\mathbb{E}[\tilde\xi^{\x,\p}(1)]}{a_1},~i=1,...,N.
    \end{align}
\end{lem}
\textit{Existence and Uniqueness:} Restricting strategies to compact sets, we can obtain existence results for the fixed point map in a projected form, using a direct application of Brouwer's theorem~\cite{rudinBaby}.
\begin{lem}
    If $\mathcal X_i,i=1,...,N$ are  convex and compact, there exists a fixed point for the equation
    $$
        \x=\mathrm{proj}_{\mathcal X}[G_{\p}(\x)].
   $$
\end{lem}

It is difficult to show uniqueness of equilibria for a general form of the PT game. One naturally expects it to display multiple equilibria. The following theorem shows that for a specific instance, where the players are symmetric and the functions $v_i$ are identical and exponential, and the vector $\p$ is restricted to be identical across players, we obtain unique equilibria at each $\p$, as well as a structural result.
\begin{thm}\label{thm:uniq-PT-sym}
    Let the players be symmetric, i.e., $a_i=a$ for all $i$, and let $\mathcal P=\{p\textbf{1}:p\in\mathbb{R}\}$.
    For all $i$, let  $v_i(z)=1-\exp(-\lambda z)$, for some $\lambda>0$. Then $|\widetilde X^*(\p)|=1$  and  $\x^*\in X^*(\p)$ satisfies $x^*_i-y_i=x^*_1-y_1$ for all players $i$.
\end{thm}
\begin{proof}
    Let $\p:=p\textbf{1}$. 
    Define
    $$
        F(\Delta)\triangleq\frac{\sum_{k=1}^Mq_k\xi_k\exp(\lambda\Delta\xi_k)}{\sum_{k=1}^Mq_k\exp(\lambda\Delta\xi_k)}.
    $$
    We have
    \begin{align*}
        \frac{\partial F(\Delta)}{\partial\Delta}&=\lambda\frac{\sum_{k=1}^Mq_k\exp(\lambda\Delta\xi_k)\sum_{k=1}^Mq_k\xi_k^2\exp(\lambda\Delta\xi_k)-\left(\sum_{k=1}^Mq_k\xi_k\exp(\lambda\Delta\xi_k)\right)^2}{\left(\sum_{k=1}^Mq_k\exp(\lambda\Delta\xi_k)\right)^2},\\
        &=\lambda\left[\frac{\sum_{k=1}^Mq_k\xi_k^2\exp(\lambda\Delta\xi_k)}{\sum_{k=1}^Mq_k\exp(\lambda\Delta\xi_k)}-\left(\frac{\sum_{k=1}^Mq_k\xi_k\exp(\lambda\Delta\xi_k)}{\sum_{k=1}^Mq_k\exp(\lambda\Delta\xi_k)}\right)^2\right],\\
        &=\lambda \left[\mathbb{E}_{q^0}[\xi^2]-(\mathbb{E}_{q^0}[\xi])^2\right],
    \end{align*}
    where $q^0$ is the distribution $(q^0(1),...,q^0(N))$ over $\Xi$, with
    $$
    q^0(i)=\frac{q_i\exp(\lambda\Delta\xi_i)}{\sum_{k=1}^Mq_k\exp(\lambda\Delta\xi_k)}.
    $$
    Since $\xi$ had positive variance under the distribution $q$, it follows that $[\mathbb{E}_{q^0}[\xi^2]-(\mathbb{E}_{q^0}[\xi])^2]>0$, and hence $F(\Delta)$ is strictly increasing in $\Delta$. 
    
    Note that, since $a_i=a$ and $p_i=p$, \eqref{eq:defn-kappa-p-x} is equivalent to
$$
        \mathbb{E}[\tilde\xi^{\x,\p}(1)]=\mathbb{E}[\tilde\xi^{\x,\p}(i)]~\forall ~i.
    $$
    Observe that $\mathbb{E}[\tilde\xi^{\x,\p}(i)]=F(y_i-x_i)$, and since $F(\Delta)$ is strictly increasing in $\Delta$, it follows that \eqref{eq:defn-kappa-p-x} is equivalent to $$
        y_1-x_1=y_j-x_j~\forall~j.$$

    Thus Lemma \ref{lem:charac-eq-PT} is now equivalent to the statement that, there exists a Nash equilibrium at $\p$ (i.e., a solution to the equation $\x=G_{\p}(\x)$) if and only if there exists
     $\beta$ such that
     
   \begin{align}\label{eqn:fix-poi-beta}
        p-\frac{2a\beta}{N}=F(\beta),
    \end{align}
    where $$\beta=y_j-x_j~\forall j.$$
    Since $p-\frac{2a\beta}{N}$ is strictly decreasing and $F(\beta)$ is strictly increasing in $\beta$, there is exactly one solution $\beta^*$ to \eqref{eqn:fix-poi-beta}. Hence we obtain the unique equilibrium at $\p$,
    $$\x^*(\p)=\y-\beta^*.$$
\end{proof}
The exponential $v$ models a uniformly risk-averse player. While in general PT transforms are not uniformly risk averse, the theorem suggests that PT transforms that are partially of the above form, may exhibit unique equilibria for some parameters.

\textit{Convergence to Equilibria:} As in the EUT case, we use Algorithm \ref{alg:con-alg-grad} in the PT case as well. The only difference is that $\Delta_i$ is replaced by
$$\tilde\Delta_i:=\frac{\partial}{\partial x_i}\widetilde J_i(x_i,\x_{-i},\p).$$
The PT game is only a concave game. In the absence of potential or monotone structure, it is not easy to establish global convergence of gradient dynamics.
However, there are conditions for local convergence to Nash equilibria~\cite{tan2023necessary}. Not all PT transformations satisfy these conditions. In Section \ref{sec:numerical-study}, we present an example as Case 1, which satisfies these conditions for local convergence. Global convergence  for more general forms of the PT game  will be taken up in a subsequent work.

 \subsection*{Preservation and Vanishing of EUT Equilibria}
 The EUT to PT transformation behaves differently with different Nash equilibria of the EUT setting. We will call an equilibrium $\x^*$ in $X^*(\p)~$\emph{symmetric} for a player $i$ if the PT transformation affects all the outcomes of that equilibrium symmetrically, i.e.,
 $$v_i(\overline J_i(x^*_i,x^*_{-i},\p,\xi_j))=v_i(\overline J_i(x^*_i,x^*_{-i},,\p,\xi_k))~\forall~ j,k.$$
 Thus, the equilibrium $\y$ is symmetric for all players at price $\overline\xi\textbf{1}$. Similarly, for $\p\ne \overline\xi\textbf{1}$, any $\x^*\in X^*(\p)$ for which $x_j^*=y_j^*$ is symmetric for player $j$. Any equilibrium which is not symmetric for player $i$, will be called \emph{asymmetric} for player $i$. We present two results, showing how a symmetric equilibrium is preserved, while asymmetric equilibria are not.
\begin{prop}\label{prop:equilibrium-PT}
    We have $\y\in\widetilde X^*(\overline\xi\textbf{1})$.
 \end{prop}
 \begin{proof}
      Let $v_i^{\prime}$  denote the first derivative of $v_i$. 
     Using  \ref{eqn:simplified-x-y-p} with $\p=\overline\xi\textbf{1}$ and $\x=\y$, this evaluates to
     \begin{align}\label{eq:deriv-for-pt-0}
        \frac{\partial}{\partial x_i}\widetilde J_i(y_i,\y_{-i},\overline\xi\textbf{1}) &= v_i^{\prime}(b_i-{\overline\xi}y_i)\sum_{j=1}^M q_j[-{\overline\xi}+\xi_j]=0.
     \end{align}
     Since $v_i$ is non decreasing and concave and $J_i$ is concave, it follows~\cite{boyd2004convex} that $\tilde u(\cdot)=\sum_jq_j[v_i\circ J_i(\cdot,\xi_j)]$ is concave. Then \eqref{eq:deriv-for-pt-0} implies \eqref{def:Nash-eqbm-PT}.
 \end{proof}
 Thus $\y$, which is a symmetric Nash equilibrium of $\mathcal G_{\text{EUT}}$ is also a Nash equilibrium of $\mathcal G_{\text{PT}}$, when $\p=\overline\xi\textbf{1}$. 

\begin{prop}\label{prop:equilibrium-destr}
    If some $v_k$ is strictly concave, then any element of  $X^*(\p)$ which is asymmetric for player $k$, will not belong to $\widetilde X^*(\p)$. 
 \end{prop}
 \begin{proof}
     Let $\p\in\mathcal P_1$. Let $\x^*\in \hat X^*(\p)$ be asymmetric for player $k$. Clearly, then $x^*_k\ne y_k$. Define $\Delta_k=y_k- x_k^*$. We present the proof for the case $\Delta_k>0$ (For $\Delta_k<0$, the proof is the similar with appropriate sign changes).  For $\p\in\mathcal P_1$, and at $\x^*$, we have
 \begin{align*}
     \frac{\partial}{\partial x_k}\widetilde J_k(x_k^*,\x_{-k}^*,\p)=\sum_{j=1}^M q_j(\xi_j-\overline\xi)v_k^{\prime}(\tilde b_k-(\xi_j-\overline\xi)\Delta_k),
 \end{align*}
where 
$$\tilde b_k=b_k-y_k\overline\xi+x_k(\overline\xi-p_k)-\frac{N^2}{4a_i}(p_i-\overline\xi)^2.$$ 
Since $v_k(z)$ is strictly concave, $v_k^{\prime}(z)$ is strictly decreasing in $z$. Thus,
\begin{align}
    v^{\prime}_k(\tilde b_k-(z-\overline\xi)\Delta_k)=v^{\prime}_k(\tilde b_k)+\epsilon(z),
\end{align}
where $\epsilon(z)$ satisfies,
\begin{align}\label{eq:epsilon-zed}
    \epsilon(z) \begin{cases}
        <0, &z<\overline\xi,\\
        >0, &z>\overline\xi.
    \end{cases}
\end{align}
Thus we have
\begin{align*}
    \frac{\partial}{\partial x_k}\widetilde J_k(x_k^*,\x_{-k}^*,\p) &=\sum_{j=1}^M q_j(\xi_j-\overline\xi)[v_k^{\prime}(\tilde b_k)+\epsilon(\xi_j)]>0,
\end{align*}
where the inequality follows by noting that $\overline\xi$ is the mean of $\xi$, and \eqref{eq:epsilon-zed}.
Thus $\x^*$ is not a Nash equilibrium.
 \end{proof}
  Thus when all  players have strictly concave PT transformations $v_i$, none of the equilibria which are asymmetric for any player, are preserved.
\begin{cor}
    If the $v_i$ are strictly concave for all $i$,then, for $\p\in\mathcal P_1$,
    $$\widetilde X^*(\p)\cap X^*(\p)=
            \{\y\}\mathds{1}_{\{\p=\overline\xi\textbf{1}\}}.$$
\end{cor}
 \begin{figure}
     \centering
     \begin{tikzpicture}[scale=1.2][domain=0:3.5]

  \draw[->] (-0.2,0) -- (4.2,0) node[right] {$x_1$};
  \draw[->] (0,-0.2) -- (0,4.2) node[above] {$x_2$};

  \draw[color=red]    (0,3) -- (3,0);
  \draw[color=blue]    (0,4) -- (4,0);
  \node at (3,-0.3) {($y_1+y_2$,0)};
  \node at (-0.9,3) {(0,$y_1+y_2$)};
  \node at (1,-0.3) {($y_1$,0)};
  \node at (-0.5,2) {(0,$y_2$)};
  \draw[dashed] (1,0)--(1,3);
  \draw[dashed] (0,2)--(2,2);
  \node at (3.3,1.5) {\textcolor{blue}{$ X^*(\p)$}};
  \node at (1.8,0.5) {\textcolor{red}{$ X^*(\overline\xi\textbf{1})$}};
  \node[circle,fill=black!20,inner sep=-2]  at (2,2) {}; 
  \node[circle,fill=black!20,inner sep=-2]  at (1,3) {}; 
  \node[circle,fill=black!20,inner sep=-2]  at (1,2) {}; 
  \node at (2.5,2) {$\x(2)$};
  \node at (1.5,3) {$\x(1)$};
  \node at (0.8,1.8) {$\x$};
\end{tikzpicture}
     \caption{Vanishing and preservation of Nash equilibria}
     \label{fig:destr-nash-eq}
 \end{figure}
 
 The preservation of symmetric equilibria and vanishing of asymmetric equilibria are illustrated below via an example.
 \begin{exmp}
     Consider $N$ players. All have linear $v_i$ except player $k$, with a strictly concave $v_k$. All $\x\in X^*(\p)$ which satisfy $x^*_k=y_k$, are elements of $\widetilde X^*(\p)$ and $\x\in X^*(\p)$ with $x^*_k\ne y_k$ are not in $\widetilde X^*(\p)$. We illustrate this for the case $N=2$ in Fig. \ref{fig:destr-nash-eq}. In the absence of any PT transformation, the set of Nash equilibria is  the blue line. If the first player has a  strictly concave $v_1$, the only equilibrium  preserved is the point $\x(1)$; similarly, if the second player has  a strictly concave $v_2$, the only equilibrium preserved is the point  $\x(2)$. If both $v_1$ and $v_2$ are strictly concave, none of the equilibria are preserved, unless $\p=\overline\xi\textbf{1}$, in which case $\x(1)=\x(2)=\x=(y_1,y_2)$, and this point is preserved. For higher $N$, we see a similar reduction of the dimension of the set of Nash equilibria, for every additional \lq irrational\rq~player.
\end{exmp}
While we assume that the $v_i$ are all concave, we see that a result analogous to Proposition \ref{prop:equilibrium-destr} can be obtained if the PT transformation was strictly convex. 
\begin{prop}\label{prop:equilibrium-destr-cvx}
    If some $v_k$ is strictly convex\footnote{The standing assumption for the rest of the results in this paper is that all $v_i$ are concave.}, then any element of  $X^*(\p)$ which is asymmetric for player $k$, will not belong to $\widetilde X^*(\p)$.
\end{prop}
\begin{proof}
    The proof is the same as for Proposition \ref{prop:equilibrium-destr}, with the change that $v_k^{\prime}$ here is strictly increasing. The appropriate changes in \eqref{eq:epsilon-zed}  lead to the result. 
\end{proof}
The strict nature of concavity or convexity is what leads to the vanishing of the asymmetric equilibria. Propositions \ref{prop:equilibrium-destr} and \ref{prop:equilibrium-destr-cvx} suggest that transformations of the EUT utilities by functions which are strictly concave for gains and strictly convex for losses, do not preserve most of the equilibria of the EUT setting. However, the PT game may still have new Nash equilibria created by the $v$ function. The existence and location of these equilibria will depend on the $v$ function, and the values of system parameters. Below, we show an example where for $\p=\overline\xi\textbf{1}$, and for a particular PT transformation function $v_i$, the PT game has no equilibria outside the set of the EUT equilibria.

\begin{exmp}\label{exmp:linear-deriv}
    For all $i$, let $v_i$ have derivative
    $v_i^{\prime}(z)=cz+d$ with
    $b_i-\overline\xi y_i<-\frac{d}{c}.$
    Then $\widetilde X^*(\overline\xi\textbf{1})\subseteq X^*(\overline\xi\textbf{1})$. Using Proposition \ref{prop:equilibrium-PT}, we see that $|\widetilde X^*(\overline\xi \textbf 1)|=1$. Thus the number of equilibria reduce from an uncountable set to a singleton, under the EUT to PT transformation.
\end{exmp}
\begin{proof}
    Say $\x^*\in\widetilde X^*(\overline\xi\textbf{1})$, but $\x^*\notin X^*(\overline\xi\textbf{1})$. Then, from \eqref{eqn:equilibrium-criterion},
    $$
        \Delta\triangleq \sum_{i=1}^N(x_i-y_i)\ne 0.
    $$
    Since $\x^*$ is a Nash equilibrium, we have
    $$
        \frac{\partial}{\partial x_i}\widetilde J_i(x_i,\x_{-i},\p) = 0~,i=1,...,N.
    $$
    This is equivalent to
    \begin{align*}
        \sum_{j=1}^Mq_j[v_i^{\prime}(J_i(x^*_i,\x^*_{-i},\p,\xi_j))(-\frac{2a_i}{N^2}\Delta+\xi_j-\overline\xi)]=0~,i=1,...,N.
    \end{align*}
    Using the fact that $v^{\prime}(z)=cz+d$, this can be simplified to, for each $i$,
    \begin{align*}
        \frac{2a_i^2}{N^4}\Delta^3-\frac{2a_i}{N^2}(b_i-\overline\xi y_i+\frac{d}{c})\Delta=(y_i-x_i)\sum_{j=1}^Mq_j(\xi_j-\overline\xi)^2.
    \end{align*}
    Summing over all $i$ and dividing by $\Delta$, we get
    \begin{align*}
        \sum_{i=1}^N \left[ \frac{2a_i^2}{N^4}\Delta^2-\frac{2a_i}{N^2}\left(b_i-\overline\xi y_i+\frac{d}{c}\right)\right]=-\sum_{j=1}^Mq_j(\xi_j-\overline\xi)^2.
    \end{align*}
    From our initial assumption on the parameters, and $\Delta$ being non zero, the LHS is strictly positive, while the RHS is negative, which is a contradiction. The result follows.

\end{proof}

\section{How Prospect Theory Shapes Market Equilibria}\label{sec:numerical-study}
We adapt the game model presented in Section \ref{sec:GameModel} to an electricity market application, with the aim to analyze the impact of users' irrationality on the market equilibria. The EUT  game corresponds to the case where the users behave rationally, and the PT game represents the case where the users' behaviour departs from the rational model. 

The coordinator is an aggregator which serves a community of $N$ consumers (users), by buying electricity on the wholesale market, and proposing two-part tariffs. During the first part of the day, the consumers can plan ahead and purchase a certain quantity of electricity. For consumer $i$, this quantity is denoted by  $x_i$; the fixed fee (price) charged by the aggregator for this purchase is $p$ per unit. The forecasted consumption of consumer $i$ is $y_i$. If the consumption of consumers exceeds $x_i$, then consumer $i$ needs to buy the $y_i-x_i>0$ from the aggregator at a variable charge of $c+\hat\xi$ where $c>0$ is a constant (e.g., peak electricity price from the wholesale market) and $\hat\xi$ is a random variable (e.g., due to the fluctuations of renewables). The variable in our model, $\xi$, is therefore equal to $c+\hat\xi$.  If the consumer $i$ uses less $y_i-x_i\leq 0$ than what it purchased, then it is refunded at the same variable charge.

The total cost paid by customer $i$ is a function that depends on their purchase, consumption, price of purchase and randomness from the renewables. It is in the interest of the consumers to have the total demand $\sum_iy_i$ to be close to the total units declared initially, $\sum_ix_i$.  We assume that the users are homogeneous, i.e., $a_i=a,b_i=b$  for all $i$. 

 We consider two candidates for the utility for the coordinator to maximize,
$$
     \overline J_0^A(\p,\x) =-||\x-\y||^2+||\p||^2,~
     \overline J_0^B(\p,\x) =\sum_{i=1}^N\overline J_i(x_i,\x_{-i},\p).
$$
 The first represents a symmetry plus revenue objective; the second represents economic efficiency.
 \subsection{EUT Game}
 Let 
 $$
     \mathcal P_1=\{p\textbf{1}:p\in\mathscr P\},
 $$
 where 
  $\mathscr P=[\underline p,\overline p]$ for some positive  $\underline p$ and $\overline p$. Using  \eqref{eq:selected-eq}, we see that 
$
     \overline J_0^A(\p,\hat\x^{*}(\p)),
 $
 for $N>2a$, is maximized at
 $$p^A_{\text{EUT}}=\mathrm{proj}_{\mathscr P}[\frac{\overline\xi}{1-\frac{4a^2}{N^2}}].$$
 Similarly we see that 
$
    \overline J_0^B(\p,\hat\x^{*}(\p)),
 $
 for $N>2$, is maximized when
$$
     p^B_{\text{EUT}}=\mathrm{proj}_{\mathscr P}[\overline\xi-\frac{a\sum_{i=1}^Ny_i}{N^2(\frac{N}{2}-1)}].
 $$
 For a large number of users (with $y_i$ bounded or with appropriate growth conditions), both $p^A$ and $p^B$ are approximately equal  to $\overline\xi$.
 \subsection{PT Game}
 In the PT case, since we lack analytical expressions for the equilibria, we will use numerical simulations to obtain insights. We assume that the distortions are symmetric, i.e., $v_i=v$. 
Let us consider a simple example, with 2 users and with $M=2$. Let $a_1=a_2=0.5$, $b_1=b_2=10$, $\xi_1=10$, $\xi_2=50$, $q_1=0.25$, $q_2=0.75$, $y_1=7$, $y_2=4$.
 We consider two forms of the distortion function $v$, and obtain Nash equilibria using a user-wise gradient ascent (Algorithm \ref{alg:con-alg-grad}), starting from random initializations.

 \noindent \textbf{Case 1}. Here $v(z)=1-\exp(-0.1z)$. From Theorem \ref{thm:uniq-PT-sym}, at each price we will have a unique equilibrium. It can also be shown, using \cite[Theorem 5]{tan2023necessary} that for each Nash equilibrium of the PT game, there is a neighbourhood  such that gradient play starting in this neighbourhood converges to this Nash equilibrium. Equilibria obtained by simulations for certain prices are depicted in Fig. \ref{fig:Nash-eq-creation-strict-ccv}, with the number next to the point indicating the price $p$.
 
 For both users, the value $x_i$ decreases as price increases. At  $p=\overline\xi=40$, the equilibrium obtained is at $(y_1,y_2)=(7,4)$. Other than this, we do not recover any EUT equilibrium, as expected from  Theorem \ref{thm:uniq-PT-sym}, and Propositions \ref{prop:equilibrium-PT} and \ref{prop:equilibrium-destr}. The equilibria obtained satisfy the structural condition obtained in Theorem \ref{thm:uniq-PT-sym},
 \begin{align}\label{eq:dashed-line}
     y_1-x_1=y_2-x_2,
 \end{align}
 which is represented by the dashed line. Numerically, with the price $p$ being restricted to the set $[1,50]$, we obtain the prices that maximize $\overline J_0^A$ and $\overline J_0^B$ as
$p^A_{\text{PT}}=\overline p=50,p^B_{\text{PT}}=\underline p=1$ respectively. The corresponding equilibria are indicated with a red circle and a black circle, respectively. Recall that in this case $p^A_{\text{EUT}}=\mathrm{proj}_{\mathscr P}[\frac{4}{3}\overline\xi]$. If $\overline p<\frac{4}{3}\overline\xi$, we will have $p^A_{\text{PT}}=p^A_{\text{EUT}}$. If $\overline p\ge \frac{4}{3}\overline\xi$, then $p^A_{\text{PT}}\ge p^A_{\text{EUT}}$. Irrational behaviour of the customers leads to a higher overall pricing by the coordinator, if they use the utility  $\overline J_0^A$. While $p^B_{\text{EUT}}$ is not available for $N=2$, we see that for small $N$, $p^B_{\text{EUT}}<\overline\xi$. The value of $\underline p$ may be even lower, which suggests that irrational preferences lower the price, if the coordinator uses the utility $\overline J_0^B$. Thus, if the coordinator uses the metric $\overline J_0^A$, which takes into account symmetry and revenue, irrationality results in an increase in prices. On the other hand, if the coordinator uses the utility $\overline J_0^B$, which is the economic efficiency, irrationality leads to price decrease.
\begin{figure}
     \centering
     \begin{tikzpicture}[scale=1]
\begin{axis}[
    axis lines=middle,
    grid=both,
    grid style={line width=.1pt, draw=gray!10},
    major grid style={line width=.2pt,draw=gray!50},
    minor tick num=4,
    xmin=-1, 
    xmax=30,
    ymin=-1, 
    ymax=30,
    x label style={at={(axis description cs:1.05,0.07)},anchor=north},
    y label style={at={(axis description cs:0.05,1)},rotate=0,anchor=south},
    xlabel={$x_1$}, 
    ylabel={$x_2$},
    label style={font=\bfseries\boldmath},
    scatter/classes={
        a={mark=square*, blue}, 
        b={mark=square*, red}, 
        c={mark=square, black}, 
        d={mark=triangle*, blue}, 
        e={mark=triangle*, red},
        f={mark=triangle*, black},
        g={mark=x, black}, 
        h={mark= diamond*, pink},
        i={mark=none}
    },
]

\addplot[
        scatter, 
        only marks,
        scatter src=explicit symbolic,
        nodes near coords*={\annotvalue},
        node near coord style={rotate=0, anchor=west, font=\scriptsize},
        visualization depends on={value \thisrow{annotation} \as \annotvalue},
    ]
    table[meta=label] {
        x       y            label  annotation
        25.16 22.17 a $~$
        23.34 20.34 a $2$
        18.95 15.95 a $5$
        15.17 12.17 a $10$
        11.75 8.75 a $20$
        9.44 6.44 a $30$
        7 4 a $40$
        2.62 -0.38 a $~$
        1.12 1.1 i $50$
        25.66 22.17 i $1$
    };
    \draw[line width=1pt,color=red] (axis cs:2.62,-0.38) circle (7);
    \draw[line width=1pt,color=black] (axis cs:25.16,22.17) circle (7);
    \addplot [dashed, domain=-1:29, mark=none] {\x-3};
\end{axis}

\end{tikzpicture}
     \caption{Nash equilibria in the  PT game, parametrized by price, for demand vector $(7,4)$ and $v(z)=1-\exp(-0.1z)$. The red  and black circles indicate equilibria for coordinator utilities $\overline J_0^A$  and  $\overline J_0^B$.}
     \label{fig:Nash-eq-creation-strict-ccv}
 \end{figure}
 \newline\noindent \textbf{Case 2}. In this case,  $v$ is given by \eqref{eq:log-v-function}. The linear portion, therefore, may lead to the existence of multiple Nash equilibria, some of which are preserved from the EUT game. Some of the  Nash equilibria we obtain are depicted in Fig. \ref{fig:Nash-eq-creation}, with the number next to the point indicating the price $p$. 
 The dashed line represents as before, equilibria that satisfy \eqref{eq:dashed-line}; while the $v_i$ are not exponential, as in Theorem \ref{thm:uniq-PT-sym}, a structural result similar to that of Theorem \ref{thm:uniq-PT-sym} can be obtained when the $v_i$ are identically logarithmic, and it can be shown that the equilibria will satisfy \eqref{eq:dashed-line}. Hence, these points are generated in regions where both users are being distorted by logarithmic preferences. The dotted lines represent the set of Nash equilibria in the EUT game with price $40$. We see that the PT game recovers some of these equilibria. It also recovers some of the the EUT equilibria for prices close to $40$ as well (not included in the figure to avoid too many close points in the same region). Since a large number of equilibria  lie on the dashed line, similar to what we had in Case 1, the behaviour of $p^A_{\text{PT}}$ is  similar to that in Case 1 as well, and we get $p^A_{\text{PT}}=\overline p=65$. However, due to $v$ having a linear region, we see that the EUT behaviour has some influence on the economic efficiency $J_0^B$.  The optimal equilibrium is at price $p^B_{\text{PT}}=5$. While $p^B_{\text{PT}}$ reduces from $p^B_{\text{EUT}}$, it does not reduce as much as in Case 1. Thus, the extent of reduction in price depends on the strict concavity of the consumer's prospect theoretic value function.

 \begin{figure}
     \centering
     \begin{tikzpicture}[scale=1]

\begin{axis}[
    axis lines=middle,
    grid=both,
    grid style={line width=.1pt, draw=gray!10},
    major grid style={line width=.2pt,draw=gray!50},
    minor tick num=3,
    xmin=-30, 
    xmax=60,
    ymin=-30, 
    ymax=60,
    x label style={at={(axis description cs:1.05,0.37)},anchor=north},
    y label style={at={(axis description cs:0.35,1)},rotate=0,anchor=south},
    xlabel={$x_1$}, 
    ylabel={$x_2$},
    label style={font=\bfseries\boldmath},
    scatter/classes={
        a={mark=square*, blue}, 
        b={mark=square*, red}, 
        c={mark=square, black}, 
        d={mark=triangle*, blue}, 
        e={mark=triangle*, red},
        f={mark=triangle*, black},
        g={mark=x, black}, 
        h={mark= diamond*, pink},
        i={mark=none}
    },
] 
\addplot[
        scatter, 
        only marks,
        scatter src=explicit symbolic,
        nodes near coords*={\annotvalue},
        node near coord style={rotate=0, anchor=west, font=\scriptsize},
        visualization depends on={value \thisrow{annotation} \as \annotvalue},
    ]
    table[meta=label] {
        x       y            label  annotation
        14.40 48.76 a $1$
        51.66 11.29 a $1$
        13.3 44.09 a $2$
        46.82 10.08 a $2$
        23.42 20.3 a $2$
        17.75 14.42 a $~$
        11.01 27.12 a $5$
        14.2 10.7 a $7$
        10.46 6.71 a $10$
        -1.26 12.26 a $~$
        0.28 10.71 a $~$
        6.43 5.47 a $~$
        4 7 a $~$
        1.71 9.28 a $40$
        -1.13 12.13 a $~$
        -1.51 12.51 a $~$
        -2.42 -1.72 a $~$
        -5.06 -4.39 a $~$
        -14.85 -14.25 a $~$
        -24.8 -24.25 a $~$
        36 -25 i $X^*(40)$
        -9.42 -1.72 i $53$
        -12.06 -4.39 i $55$
        -21.35 -14.25 i $60$
        -28 -20 i $65$
        18.75 14.42 i $5$
    };
    \draw[line width=1pt,color=red] (axis cs:-24.8,-24.25) circle (17);
    \draw[line width=1pt,color=black] (axis cs:17.75,14.42) circle (17);
    \node[rotate=40] at (axis cs:39,40) {\scriptsize $x_1-x_2=3$};
    \addplot [color=blue,dotted, domain=-30:60, mark=none] {11-\x};
    \addplot [dashed, domain=-25:55, mark=none] {\x-3};

\end{axis}

\end{tikzpicture}
     \caption{Nash equilibria in the  PT game, parametrized by price, for demand vector $(7,4)$ and $v$ given by \eqref{eq:log-v-function}. The red  and black circles indicate  equilibria for  coordinator  utilities $\overline J_0^A$  and  $\overline J_0^B$.}
     \label{fig:Nash-eq-creation}
 \end{figure}

\section{Conclusion}
We studied how the Nash equilibria of an aggregative coordination game are transformed when the players' irrationality is included, using prospect theory. We proved that for specific choice of value functions, symmetric equilibria are preserved while asymmetric equilibria are not. In many cases, the prospect theoretic transform reduces the size or dimension of the existing set of Nash equilibria. Further, we learned prospect theoretic Nash equilibria in a stylized electricity market using a distributed gradient play algorithm where an $N+1$-player updates the price to steer the users to a desired equilibrium. In particular, we show numerically that irrationality can drive fixed prices up or down, depending on the $N+1$-player's selection function. While we obtained results on the structure and uniqueness of PT equilibria in specific settings, a more general analytical characterization of newly created prospect theoretic equilibria will be considered in an extension, along with theoretical guarantees for the global convergence of the gradient play (or other) algorithms.
 
\bibliographystyle{ieeetr}
\bibliography{sample}
    
\end{document}